\documentclass[a4paper,USenglish,cleveref,autoref,thm-restate,numberwithinsect]{lipics-v2021}

\usepackage{xparse,ifthen,amssymb,empheq,hyperref,mathrsfs,esint,mathtools,systeme,mathabx,accents,cancel,physics,listings,beramono,thmtools}
\newcommand\q\qty{}

\newcommand{\set}[1]{\left\{ #1 \right\}}
\newcommand{\N}{\mathbb{N}}
\newcommand{\inv}[1]{\frac{1}{#1}}

\NewDocumentCommand{\E}{e_ o }{\IfValueTF{#1}{\mathbb{E}_{#1}}{\mathbb{E}}\IfValueT{#2}{\q[#2]}}
\RenewDocumentCommand{\P}{e_ o }{\IfValueTF{#1}{\mathbb{P}_{#1}}{\mathbb{P}}\IfValueT{#2}{\q[#2]}}

\newcommand{\onethird}{\frac{1}{3}}
\newcommand{\twothirds}{\frac{2}{3}}

\newcommand{\RYES}{{R_{\mathrm{YES}}}}
\newcommand{\RNO}{{R_{\mathrm{NO}}}}
\newcommand{\LYES}{{L_{\mathrm{YES}}}}
\newcommand{\LNO}{{L_{\mathrm{NO}}}}

\newcommand{\defsf}[1]{\csdef{#1}{\mathsf{#1}}}
\defsf{leak}
\defsf{fL}
\defsf{fBPL}
\defsf{L}
\defsf{BPP}
\defsf{pad}
\defsf{unpad}
\newcommand{\BPdL}{\mathsf{BP{\boldsymbol{\cdot}}L}}
\defsf{prL}
\defsf{prP}
\defsf{prBPP}
\newcommand{\prBPdL}{\mathsf{prBP{\boldsymbol{\cdot}}L}}
\defsf{GenDesign}
\defsf{poly}
\newcommand{\I}{\mathcal{I}}
\newcommand{\U}{\mathcal{U}}

\title{Leakage-Resilient Hardness Equivalence to Logspace Derandomization}
\author{Yakov Shalunov}
{University of Chicago\footnote{Research conducted as a student at the California Institute of Technology.}}
{yasha@uchicago.edu}
{}
{}

\authorrunning{Yakov Shalunov}

\Copyright{Yakov Shalunov} 



\relatedversion{} 
\relatedversiondetails[linktext={\url{https://arxiv.org/abs/2312.14023}}]{Previous version}

\funding{Funded by Chris Umans' Simons Investigator Grant and Dr. Arjun Bansal and Ms. Ria Langheim's donation through the California Institute of Technology's ``Summer Undergraduate Research Fellowship'' (SURF) program.}

\acknowledgements{Most of all, I'd like to thank Professor Chris Umans at Caltech for providing guidance and suggestions throughout this project, helping me find ways forward when I was stuck, and providing references to relevant resources. I'm also very grateful to Winter Pearson, member of the undergraduate class of 2024 at Caltech, for helping with copy-editing and proofreading.}

\nolinenumbers 

\begin{document}

\maketitle

\DeclareDocumentCommand{\binstr}{e^}{\set{0,1}^{\IfValueTF{#1}{#1}{*}}}

\begin{abstract}
    Efficient derandomization has long been a goal in complexity theory, and a major recent result by Yanyi Liu and Rafael Pass identifies a new class of hardness assumption under which it is possible to perform time-bounded derandomization efficiently: that of ``leakage-resilient hardness.'' They identify a specific form of this assumption which is \emph{equivalent} to $\prP = \prBPP$.

    In this paper, we pursue an equivalence to derandomization of $\prBPdL$ (logspace promise problems with two-way randomness) through techniques analogous to Liu and Pass.

    We are able to obtain an equivalence between a similar ``leakage-resilient hardness'' assumption and a slightly stronger statement than derandomization of $\prBPdL$, that of finding ``non-no'' instances of ``promise search problems.''
\end{abstract}

\section{Introduction}
In a time-bounded setting, pseudorandom generators have historically been used to show that certain circuit-complexity lower bounds imply that $\mathsf{P}=\BPP$ (in particular, problems in $\mathsf{E}$ which require exponential size circuits ~\cite{impagliazzo_1997}). Another key result showed that $\mathsf{P}=\BPP$ (and specifically, a deterministic polynomial-time algorithm for polynomial identity testing) implies that either $\mathsf{NEXP}$ requires super-polynomial boolean circuits or computing permanents requires super-polynomial arithmetic circuits~\cite{kabanets_2004}. However, historically the lower bounds known to imply $\mathsf{P}=\BPP$ and those implied by $\mathsf{P}=\BPP$ have not matched. Seeing these results, one might hope to find a hardness assumption $H$ such that
\begin{center}
    ``\textit{$\mathsf{P} = \mathsf{BPP}$ if and only if there exists a problem satisfying $H$}.''
\end{center}

A recent exciting paper, ``Leakage-Resilient Hardness v.s. Randomness''~\cite{liu_2022} by Liu and Pass, has done exactly this, building on prior work by Chen and Tell~\cite{chen_2021}; Nisan and Wigderson~\cite{nisan_1994}; Impagliazzo and Wigderson~\cite{impagliazzo_1997}; Sudan, Trevisan, and Vadhan~\cite{sudan_2001}; and Goldreich~\cite{goldreich_2011}. Liu and Pass found an alternate form of hardness assumption which they were able to formulate as equivalent to $\mathsf{prP}=\mathsf{prBPP}$ using a cryptographic notion of ``leakage-resilient hard functions''~\cite{rivest_1985}—i.e., hard functions which are still uniformly hard even if you have some ``leaked'' information about the output (formalized as a ``leakage function'' of bounded output length taking the hard function's output as an input).

In the space-bounded derandomization setting, there's been an enormous effort to prove unconditionally that $\L=\mathsf{BPL}$, and while it's believed that there's no obstruction to being able to do so, this goal continues to remain elusive. However, it is also known that it is possible to get a conditional derandomization in logarithmic space with PRG machinery similar to the time-bounded randomization setting~\cite{klivans_2002}.

\subsection{Overview}

In the vein of conditional derandomization, this paper constructs a result similar to Liu and Pass for logarithmic space-bounded computation. We obtain an equivalence utilizing a space-bounded analog to the time-bounded assumption in Liu and Pass's paper.

Liu and Pass's ideas do not carry directly to the space-bounded setting because certain manipulations required by the proof cannot be carried out in logspace. In particular:

\subparagraph*{Worst-case vs. average-case hardness.}
A key step in the Liu and Pass proof that ``hardness implies derandomization'' uses error-correcting codes to convert worst-case hardness to average-case hardness. In the space-bounded setting, we do not have access to error-correcting codes capable of list-decoding the available $\frac{1}{2} + \mathcal{O}\left(\inv{n}\right)$ fraction of correct bits. As a consequence, our version of leakage-resilient hardness is average case instead of worst case. Using a stronger hardness assumption then required additional work in the converse direction.

\subparagraph*{Average-case hardness.}
When analyzing the constructed Nisan-Wigderson generator in the standard way, Liu and Pass use rejection sampling to find a prefix fixing for Yao's distingiusher-to-predictor transformation with a sublinear number of trials. However, in the space-bounded setting, we cannot write down a single fixing (since the size of the fixing is superlogarithmic), so our hardness assumption must be strengthened until even a linearly small chance of having the correct fixing contradicts the hardness assumption. This version then allows us to sample once and write it directly to the output tape without verification.

\subparagraph*{Search problem reduction.}
When Liu and Pass prove that $\mathsf{prP}=\mathsf{prBPP}$ implies their hardness assumption, they rely on a prior result by Goldreich to reduce the task of finding a non-no solution for a $\mathsf{prBPP}$ search problem to a polynomial number of $\mathsf{prBPP}$ decision problems~\cite{goldreich_2011}, which can then be solved deterministically by assumption. In the space-bounded setting, we've found no equivalent transformation and, as a result, our main theorem works directly with the non-standard derandomization of search problems (which in polynomial time is equivalent to $\mathsf{prP} = \mathsf{prBPP}$). Specifically, we use a ``one-sided search'' which, given a ``promise search problem,'' finds a solution which is not a no-instance (these terms are more precisely defined in the following section).

Nevertheless, we are able to obtain:
\begin{theorem}[Main Theorem (informal)]
    There is a logspace-computable function which is leakage-resilient hard if and only if it is possible, in logspace, to deterministically solve a one-sided search for any $\prBPdL$ search problem.
\end{theorem}
This shows a meaningful equivalence between a hardness assumption and a derandomization of a search version of $\prBPdL$. Notably, the derandomization necessary follows from the existence of a standard pseudorandom generator construction powerful enough to derandomize $\prBPdL$.

Shortly after the original release of this paper as a preprint, Pass and Renard released their own, independent result ``Characterizing the Power of (Persistent) Randomness in Log-space''~\cite{pass_2023} which presents a theorem very similar to our own---their result is a full derandomization of $\prBPdL$ search problems under a roughly equivalent condition. Their techniques in both directions are analogous to ours and their result is fundamentally similar. 

In particular, despite presenting their result differently, Pass and Renard's work demonstrates the same one-sided search as ours. While their derandomization direction (claims 5.3, 4.12, 4.9, and 4.10 in~\cite{pass_2023}) relies on a non-promise version of $\BPdL$ search problems (where non-no instances are yes instances), their hardness direction relies on promise search like ours. (Notably, Goldreich's reduction in polynomial time also produces non-no instances rather than yes instances, suggesting this is a more fundamental limit of this formulation of search problems.) Unfortunately, this produces a mismatch that makes their results not quite correct-as-written.

Pass and Renard's paper further contributes a ``partial derandomization'' analogous to Liu and Pass's low-end regime (i.e. for positive $\gamma$, containment of $\mathsf{search}\BPdL$ in $\mathsf{searchL}^{1 + \gamma}$ is equivalent to existence of $\mathcal{O}(\log^{1 + \gamma}n)$ computable hard $f$). Though we believe it contains the same mismatch of directions, after some simple reformatting into a one-sided derandomization, it will be an elegant extension to the equivalence.

We believe our proof of the key equivalence is more direct, briefer, and clearer as it, in particular, achieves equivalence without the use of Doron and Tell's~\cite{doron_2023} recent results in logspace PRGs (with no significant increase in complexity) by using a better suited version of leakage resilient hardness. Their result does, however, show that our versions of leakage resilient hardness are equivalent to each other.

Because our result follows a similar path to Liu and Pass's, in section 2 we provide an accessible summary and informal description of the key ideas in their result.

In section 3, we provide a formal statement of our result and the definitions necessary for it.

In sections 4 and 5, we prove the forward and reverse directions of the claim respectively, and in section 6 we discuss potential further work.

\section{Informal description of Liu and Pass's polynomial-time equivalence result}
In this section, we provide a high-level summary of Liu and Pass's proof, which also serves as an outline for our own work.

Since it is central to the proof, we start with the definition of ``leakage-resilient hardness,'' which is a uniform hardness assumption coined by Rivest and Shamir in 1985~\cite{rivest_1985} and modified by Liu and Pass based on Chen and Tell's~\cite{chen_2021} work with almost-all input hardness. 

In the proof, leakage is used to uniformly substitute for non-uniform hard-coded bits in the proof of correctness of the Nisan-Wigderson generator, which we expand on later in the section.

\begin{definition}[Almost-all input leakage-resilient hardness~\cite{liu_2022}\cite{rivest_1985}\cite{chen_2021}]
    Let $f: \binstr^n \to \binstr^n$ be a function. We say that $f$ is \emph{almost-all-input $(T, \ell)$-leakage-resilient hard} if for all $T$-time probabilistic algorithms $(\leak, A)$ satisfying $|\leak(x, f(x))| \leq \ell(|x|)$, for all sufficiently long strings $x$, $A(x, \leak(x, f(x))) \neq f(x)$ with probability at least  $\frac{2}{3}$.
\end{definition}

To illustrate this property, consider when $f$ is the classical example of a (potentially) hard function used in cryptography: prime factorization of a product, $x$, of two primes. Given $f(x)$ (a correct factorization of $x$), a leakage function $\leak(x, f(x))$ which picks the smaller prime number from $f(x)$ and leaks it would make prime factorization ``easy'' as it would allow the attacker $A$ to simply be a division algorithm. Since both division and identifying the smaller factor can be done in polynomial (say, $n^k$) time, and the smaller factor must be at most $n/2$ bits, prime factorization is \emph{not} $(n^k, n/2)$-leakage-resilient hard---even if, without leakage, it requires superpolynomial time.

The simplest expression of Liu and Pass's result as relevant to ours is the following:
\begin{theorem}[Liu and Pass~\cite{liu_2022}]
    There exists a constant $c$ such that for all $\varepsilon \in (0, 1)$, the following are equivalent:
    \begin{itemize}
        \item There exists a function $f: \binstr^n \to \binstr^n$, computable in deterministic polynomial time, which is almost-all input $(n^c, n^\varepsilon)$-leakage-resilient hard; and
        \item $\mathsf{prP} = \mathsf{prBPP}$.
    \end{itemize}
\end{theorem}

\subsection{Forward direction (Liu and Pass)}
Liu and Pass's forward direction---hardness assumption to $\mathsf{prP} = \mathsf{prBPP}$---adapts the original proof that the Nisan-Wigderson pseudorandom generator (NW PRG) yields derandomization, with a few key differences. The generator they construct (invented by Goldreich~\cite{goldreich_2011}) is ``targeted,'' meaning both it and the distinguishers it beats get access to a ``target'' string as an input in addition to the seed. When performing derandomization using the generator, the appropriately padded input becomes the target. (The utility of this is explained at the end of this subsection.)

Their construction of a NW-like PRG uses the error-correcting encoded output of the leakage-resilient hard function $f$ on the ``target'' $x$ instead of the standard hard-coded truth table of a non-uniformly hard function. (In our proof, the error-correcting encoding step is replaced with average-case hardness.)

Recall that in the standard proof of correctness for a Nisan-Wigderson pseudorandom generator, when converting a distinguisher to a predictor, there is a step where a prefix is fixed and a small table of a subset of outputs from the hard truth-table is hard-coded. The key difference here is that, in order to obtain a contradiction, instead of hard-coding this prefix and table non-uniformly, it is instead computed by the $\leak$ function, which has access to the hard truth-table (since it is efficiently computed from the output of $f$, which $\leak$ has access to). The prefixes are found by rejection sampling to pick a ``good one.'' (In our proof, we select one randomly.) This output can be made polynomially small compared to $n$ to satisfy the length bound $\ell$ on $\leak$. 

The final step of the proof is to argue that if there is a problem in $\prBPP$ that cannot be derandomized by this PRG, we can turn this into a distinguisher by fixing the problem. Observe, however, that for a contradiction, we need derandomization to fail in infinitely many cases (since the hardness assumption is \emph{almost}-all input). Furthermore, we are only fooling \emph{uniform} distinguishers. This is where the targeting of the PRG is necessary.

By using (a padded version of) the input to the decision problem, $x$, as the target, we ensure that we can uniformly ``find'' the infinitely many values of $x$ where the distinguisher succeeds. If the PRG were untargeted, we would have no way to uniformly identify which value of $x$ the PRG fails on for a given length, so failing to derandomize some decision problem would not guarantee a distinguisher.

(In our proof, we additionally perform a one-sided search derandomization. The search problem's search algorithm (the ``finder'') is run a with pseudorandom string generated for each possible seed until one is found which the verifier accepts. We show that such a seed must exist and that this process produces the desired derandomization. In the polynomial time regime, this one-sided search derandomization follows from decision derandomization directly without further access to a PRG and so happens on the converse side.)

\subsection{Backward direction (Liu and Pass)}
The key idea of the backward direction---$\mathsf{prP} = \mathsf{prBPP}$ to hardness assumption---is that, if we were able to assign a ``random'' string to each input, we could build such an $f$ since leaking some information about that random string would not help you reconstruct the rest of the output.

Liu and Pass make use of a result by Goldreich~\cite{goldreich_2011} which states that finding non-no solutions to a $\prBPP$ search problem can be reduced to a polynomial number of $\prBPP$ decision problems. Liu and Pass then formulate finding a ``good output'' $f(x)$ for any given input $x$ as a $\prBPP$ search problem (defined below). They then apply Goldreich's result and then the derandomization assumption to create a hard function whose output is deterministic but behaves sufficiently like a ``random string'' would.

In Goldreich's formulation, a $\prBPP$ search problem is defined as follows: 

\begin{definition}[$\prBPP$ search problems~\cite{goldreich_2011}]
$\RYES, \RNO \subseteq \binstr \times \binstr$ where $\RYES \cap \RNO = \varnothing$ represent a $\prBPP$ search problem if both of the following hold:
\begin{itemize}
    \item $\RYES$ and $\RNO$ represent a $\prBPP$ decision problem. I.e., there is a PPT machine $V$ (the ``verifier'') such that
    $$\forall (x,y) \in \RYES,~\P[V(x,y) = 1] \geq \twothirds \qand \forall (x,y) \in \RNO,~\P[V(x,y) = 0] \geq \twothirds$$
    \item There exists a PPT machine $F$ (the ``finder'') such that for all $x$, if $\exists y, (x,y) \in \RYES$, then
    $$\P[(x,F(x)) \in \RYES] \geq \twothirds$$
\end{itemize}

Let $S_R = \set{x : \exists y, (x, y) \in \RYES}$ be the set of inputs for which a solution exists.
\end{definition}

Liu and Pass formulate computing $f$ as finding solutions to a search problem. A given pair $(x, y)$ is a yes-instance of the problem if every pair $(A, \leak)$ of attacker and leakage functions whose program descriptions have length at most $\log n$ and which run in the desired time bound has $A(x, \leak(x, y)) = y$ with probability at most $1/6$. Similarly, $(x, y)$ is a no-instance if there exists $(A, \leak)$ such that $A(x, \leak(x, y)) = y$ with probability at least $1/3$. (The gap between yes- and no-instance probabilities is allowed by it being a \emph{promise} search problem and is used by the verifier.)

Only attacker pairs bounded to length $2 \log n$ need to be considered because this only excludes \emph{any particular} machine from consideration for a finite number of inputs $x$.

When allowed to use randomness, the finder is simple---as established, simply ignoring $x$ and picking a random string $y$ works with high probability. The verifier can check whether a yes-or-no instance is a yes-instance by simply iterating over all attackers and leakage functions up to length $\log n$ (truncating their run time to the required bound) and checking whether any of them compute $y$ with probability greater than $1/4$. This runs in polynomial time since there are $n$ machines each truncated to a polynomial time bound and each pair only needs to be run polynomially many times to be correct with high probability.

The desired hard function can then be computed by using Goldreich's reduction to reduce the search problem to a polynomial number of $\prBPP$ problems. These are then derandomized to deterministically produce a non-no instance. The $y$-component of the resulting non-no instance is then exactly the desired hard output, by the construction of the search problem.

(In our case, we need access to the PRG rather than merely $\prL = \prBPdL$ to perform the one-sided search, so it is performed when proving the forward direction.)

\section{Results}

\begin{theorem}[Main theorem]
    \label{thm:main}
    There exist parameters $c_1, c_2, \alpha > 0$ (in particular, $\alpha = \onethird$ works) such that the following are equivalent:
    \begin{enumerate}
        \item There exists a logspace function $f$ which is almost-all-input $(n^{c_1}, c_2 \log n, \ell, \delta)$-leakage-resilient average hard where $\ell = n^\varepsilon$ (for $\varepsilon = 2\alpha + \frac{\alpha^3}{5}$) and $\delta = \q(\frac{1}{2} - \inv{m^2})n$ (for $m = n^{\alpha^3/5}$).
        \item For any $\prBPdL$ search problem $\RYES, \RNO$, we can create a deterministic algorithm $F'$ such that for any $x \in S_R$, we have that $(x,F'(x)) \notin \RNO$. (In particular, this implies $\prBPdL = \prL$.)
    \end{enumerate}
\end{theorem}
\begin{proof}
    The forward direction is theorem \ref{thm:existence-to-derandomization} and the reverse direction is theorem \ref{thm:derandomization-to-existence}.
\end{proof}

\subsection{Preliminaries}
Note that we are working with a read-\emph{only} random tape with two-way movement. For decision problems, this definition corresponds to the complexity class $\BPdL$, though we work with its promise and search promise counterparts.

Our result is formulated in terms of ``search problems,'' analogous to the $\prBPP$ search problems defined by Goldreich in~\cite{goldreich_2011}.

\begin{definition}[$\prBPdL$ search problems]
A pair of sets $\RYES, \RNO \subseteq \binstr \times \binstr$ where $\RYES \cap \RNO = \varnothing$  represent a $\prBPdL$ search problem if both of the following hold
\begin{itemize}
    \item $\RYES$ and $\RNO$ represent a $\prBPdL$ decision problem. I.e., there is a PPT logspace machine $V$ (the ``verifier'') such that
    $$\forall (x,y) \in \RYES,~\P[V(x,y) = 1] \geq \twothirds \qand \forall (x,y) \in \RNO,~\P[V(x,y) = 0] \geq \twothirds$$
    \item There exists a PPT logspace machine $F$ (the ``finder'') such that for all $x$, if $\exists y, (x,y) \in \RYES$, then
    $$\P[(x,F(x)) \in \RYES] \geq \twothirds$$
\end{itemize}

Let $S_R = \set{x : \exists y, (x, y) \in \RYES}$ be the set of inputs for which a solution exists.
\end{definition}

This definition is identical to the polynomial variant with the added constraint that the machines are bounded to logarithmic space in addition to the polynomial time bound.

Our hardness assumption is a modification of leakage resilient hardness as defined in definition 2.1 of~\cite{liu_2022}, which in turn builds on~\cite{rivest_1985}.

\begin{definition}[Almost-all-input leakage resilient average hard]
    \label{def:leakage-average-hard}
    Let $f: \binstr^n \to \binstr^n$ be a multi-output function. $f$ is \emph{almost-all-input $(T, S, \ell, d)$-leakage resilient average hard} if for all $S$-space $T$-time probabilistic algorithms $(A, \leak)$ satisfying $|\leak(x, f(x))| \leq \ell(n)$, we have that for all sufficiently long strings $x$, 
    
    $$\P[d_H(A(x, \leak(x, f(x))), f(x)) < d(n)] < \inv{n}$$

    with probability over the internal randomness of the algorithms, where $d_H$ is Hamming distance.
\end{definition}

Note that this is ``stronger'' than a direct analog to the time-bounded assumption~\cite{liu_2022} in two ways: first, and most importantly, we've added another parameter which parameterizes how ``good of an approximation'' is allowed, whereas in the time-bounded setting, forbidding strict equality was sufficient. Second, the allowed probability of ``success'' is $\inv{n}$ rather than $\inv{3}$.

The first difference converts the worst-case hardness to average-case hardness. It is necessary because in the time-bounded setting, we have access to much more powerful error-correcting codes which are not available in the logspace setting. In the time-bounded setting, error-correcting codes let us convert a small edge in breaking the generator to an exact computation of the hard function. Meanwhile, in the logspace setting, we're limited to just that small edge, so we bake it into the hardness assumption instead.

The second difference exists to handle issues in the distinguisher to predictor transformation. In polynomial time, the attacker functions can ``try'' $o(1/n)$ strings and figure out which one is the ideal choice, but because the strings are super-logarithmic in length, we cannot verify them in logspace. Thus, the hardness assumption is strengthened to allow simply choosing a random one without verification.

While much of the high-level structure of the following proof is similar to that presented in section 2, we have that:
\begin{itemize}
    \item in the ``hard function implies derandomization'' direction, we show the desired derandomization of search problems; and
    \item in the ``derandomization implies hard function'' direction, the search problem is constructed such that anything which is not a no-instance is a valid value for the function to take for that $x$. Then, instead of Goldreich's reduction, we apply the derandomization reduction directly.
\end{itemize}

\section{Existence of hard function implies derandomization}
\begin{theorem}[Existence implies derandomization]
    \label{thm:existence-to-derandomization}
    The forward direction of the main theorem (\ref{thm:main}). I.e., (1) implies (2).
\end{theorem}
\begin{proof}
    Apply lemma \ref{lemma:hard-prg} with $C = 3$ to get the necessary PRG, then apply lemma \ref{lemma:prg-derand} to perform the desired derandomization with the PRG.
\end{proof}

\newcommand{\distinguish}{D_x^m}
\newcommand{\generator}[1][x]{G_{#1}^m}
The following are fairly standard definitions of a distinguisher and pseudorandom generator with the caveat of the additional ``target'' string. For brevity, write $\distinguish(\gamma)$ for $D(1^m, x, \gamma)$ and similarly for $G$.

\begin{definition}[Targeted distinguisher]
    Given a function $G: \set{1}^m \times \binstr^n \times \binstr^d \to \binstr^m$, a targeted distinguisher for $G$ with advantage $\beta > 0$ is a machine $D: \set{1}^m \times \binstr^n \times \binstr^m$ such that for all sufficiently large $m$ and for all $x \in \binstr^n$, we have that
    \begin{equation}
        \label{eqn:distinguisher}
        \abs\Big{\P_{s \sim \U_{d}}\q\big[\distinguish(\generator(s)) = 1] - \P_{\gamma \sim \U_m}\q\big[\distinguish(\gamma) = 1]} \geq \beta
    \end{equation}
    where $\U_k$ is the uniform distribution on $\binstr^k$.

    Throughout, we will refer to targeted distinguishers as ``distinguishers.''
\end{definition}

This definition becomes useful in the context of the following (modified from~\cite{goldreich_2011}\cite{liu_2022}):

\begin{definition}[Uniform targeted pseudorandom generator]
    For parameters $S$, $n$, and $d$ dependent on output length $m$, an $S$-secure uniform $(n, d)$-targeted PRG is a deterministic function $$G: \set{1}^m \times \binstr^n \times \binstr^d \to \binstr^m$$ which satisfies the property that there does not exist a targeted distinguisher with advantage $\beta$ running in space $S$ for any $\beta > 0$.

    Throughout the rest of this article, ``PRG'' and ``pseudorandom generator'' refer to a uniform targeted PRG. $n$ is the ``target length'' of the generator and $d$ is the ``seed length.''
\end{definition}

We show that our hardness assumption implies an $\mathcal{O}(\log m)$-secure $(\poly(m), \mathcal{O}(\log m))$ PRG and then that the existence of an such a PRG implies the desired derandomization.

\subsection{Hard function implies existence of PRG}
First, we prove that a hard function with the given parameters implies the existence of a PRG.

Fix a sufficiently small value of $\alpha > 0$. $\alpha = 1/3$ works.

\begin{lemma}[Hard function implies existence of PRG]
    \label{lemma:hard-prg}
    For any $C$, there exist constants $c_1, c_2$ such that if there exists a logspace computable function $f$ such that $f$ is almost-all-input \\$\q(n^{c_1}, c_2\log n, n^\varepsilon, \q(\frac{1}{2} - \frac{1}{m^2})n)$-leakage resilient average hard with $m = n^\frac{\alpha^3}{5}$ and $\varepsilon = 2\alpha + \frac{\alpha^3}{5}$, then there exists a $C \log m$-secure $(n,\frac{\log n}{\alpha})$ PRG which is logspace computable.
\end{lemma}

As is standard for a Nisan-Wigderson generator, we will require designs for the construction.
\begin{definition}[Combinatorial design]
    A design is a collection of (potentially large) sets with a bounded pairwise intersection size. More precisely:

    For any natural numbers $(d,r,s)$ such that $d>r>s$, a $(d,r,s)$-design of size $m$ is a collection $\I = \set{I_1,\dots,I_m}$ of subsets of $[d]$ such that for each $j \in [m]$, $|I_j| = r$ and for each $k\in [m]$ such that $k \neq j$, $|I_j \cap I_k| \leq s$.
\end{definition}
Because we're working in logspace, we will need the following lemma:
\begin{lemma}[Creating designs]
    There is a deterministic algorithm we call $\GenDesign(d, \alpha)$ which produces a $(d, \alpha d, 2\alpha^2d)$ design of size $2^{\frac{\alpha^4d}{5}}$ in space $S_\GenDesign = \mathcal{O}(d)$.
\end{lemma}
\begin{proof}
    Proven by Klivans and Melkebeek in the appendix of~\cite{klivans_2002}.
\end{proof}

\begin{proof}[Outline of proof of \ref{lemma:hard-prg}]
As presented in the overview, the construction of the PRG is a fairly standard Nisan-Wigderson generator~\cite{nisan_1994} construction, with the caveats that it is a targeted PRG and that instead of the truth table of a hard-coded non-uniformly hard function, we use the output of our $f$ on the target string. The key change is that the leakage function then replaces the hard-coded tables in the search-to-decision reduction. 

We show that this can be done in logspace and that the hard-coded tables fit in the inverse polynomial leakage bound $n^\varepsilon$. In order to compute the tables in logspace, the probability of generating ``good tables'' needs to be low because it's impossible to verify the tables. However, a probability high enough to contradict leakage-resilient hardness, i.e., over $\inv{n}$, is achievable.

Assuming a distinguisher that defeats the PRG, the attacker function is then able to use the tables from the leakage function together with the distinguisher to perform the search-to-decision reduction with enough accuracy to violate the average-case hardness constraint.
\end{proof}

\begin{proof}[Proof of \ref{lemma:hard-prg}]
With an appropriate choice of $c_1, c_2$, we can construct a $C \log m$-secure PRG for any constant $C \geq 0$.

We construct a targeted PRG $G: \set{1}^m \times \binstr^n \times \binstr^d \to \binstr^m$, where $n = m^\frac{5}{\alpha^3}$ and $d = \frac{\log n}{\alpha}$.

The algorithm is similar to that in time-bounded case~\cite{liu_2022}: it is simply a Nisan Wigderson generator~\cite{nisan_1994}, with the hard-coded truth table replaced with one outputted by the hard function.

On input $1^m, x, s$, it proceeds as follows:

\begin{enumerate}
    \item Compute $z = f(x)$. Define $h(i) = z_i$.
    \item Compute $\I = \GenDesign(d, \alpha)$. Note that $r = d \alpha = \log n$ and $2^{\frac{\alpha^4d}{5}} \geq m$ due to our choice of $n$ and $d$.

    \item Output
    $$G(1^m,x,s) = h(s_{I_1}) \cdots h(s_{I_m})$$
\end{enumerate}

Suppose there exists a distinguisher $D$ computable in space $C \log m$ with advantage $\beta > 0$. We will use this to approximate $f$ with enough of an advantage to violate the average hardness assumption by a standard hybrid argument (lemma 3.15 in~\cite{liu_2022}, proposition 7.16 in~\cite{vadhan_2012}, theorem 10.12 in~\cite{arora_2009}).

We can remove the absolute value from the distinguisher definition \eqref{eqn:distinguisher} by noting that there exists $b \in \set{0,1}$ such that

$$\P_{s \sim \U_{d}}[\distinguish(\generator(s)) = b] - \P_{\gamma \sim \U_m}[\distinguish(\gamma) = b] \geq \beta$$

since flipping $b$ flips the sign.

Then, for every $j \in \set{0,\dots, m}$, define

$$H_j = (h(s_{I_1}),\dots,h(s_{I_j}),w_{j+1},\dots,w_m)$$

(with $h$ defined as in the generator) where $s \sim \U_d$ and each $w_k \sim U_1$ (for $j + 1 \leq k \leq m$). Notice that $H_0 = \U_m$ and $H_m = G(1^m, x, \U_d)$. Therefore it follows that

\begin{multline*}
    \inv{m} \sum_{j \in 1}^m \q\Big(\P_{y,w}[\distinguish(H_j) = b] - \P_{y,w}[\distinguish(H_{j-1}) = b]) 
    \\=\inv{m}\q\Big(\P_{y,w}[\distinguish(H_m) = b] - \P_{y,w}[\distinguish(\U_0) = b])
    \geq \frac{\beta}{m}
\end{multline*}

Considering $j$ as a random variable distributed uniformly over $[m]$, we get that

$$\E_{j \in [m]}\Big[\P_{y,w}[\distinguish(H_j) = b] - \P_{y,w}[\distinguish(H_{j-1}) = b]\Big] \geq \frac{\beta}{m}$$

Since $\P_{y,w}[\distinguish(H_j) = b] - \P_{y,w}[\distinguish(H_{j-1}) = b]$ is upper bounded by 1, by an averaging argument, with probability at least $\frac{\beta}{2m}$ over the choice of $j \gets [m]$, $y_{[d]\setminus I_j} \gets \binstr^{d-r}$, and $w_{[m]\setminus[j]} \gets \binstr^{m - j}$, the strings $j$, $y_{[d]\setminus I_j}$, and $w_{[m]\setminus[j]}$ will satisfy:

\begin{equation}
    \label{eq:good_choice}
    \P_{y_{I_j}\sim\U_r}[\distinguish(H_j) = b] - \P_{(y_{I_j},w_j)\sim \U_{r + 1}}[D(H_{j-1})=b] \geq \frac{\beta}{2m}
\end{equation}

Suppose we have a choice $j$, $y_{[d]\setminus I_j}$, and $w_{[m]\setminus[j]}$ satisfying the above equation. By Yao's prediction vs. indistinguishability theorem~\cite{yao_1982}, it holds that

$$\P_{(y_{I_j},w_j)\sim \U_{r + 1}}[\distinguish(H_{j-1}) \oplus b \oplus w_j = h(y_{I_j})] \geq \frac{1}{2} + \frac{\beta}{2m}$$

There then must exist a choice of $w_j$ such that

\begin{equation}
    \label{eq:good_approximation}
    \P_{y_{I_j}\sim \U_r}[\distinguish(H_{j-1}) \oplus b \oplus w_j = h(y_{I_j})] \geq \frac{1}{2} + \frac{\beta}{2m}
\end{equation}

Observe that in order to compute $H_{j-1}$ as a function of $y_{I_j}$, one does not need to know the entire truth table of $h$. Instead, note that the overlap between $I_j$ and $I_i$, $i < j$, is at most $2\alpha^2d = 2\alpha r$. Thus, as a function of $y_{I_j}$, $h(y_{I_i})$ can only take $2^{2\alpha r}$ possible values. Thus, to compute all of $H_{j-1}$, one needs $j2^{2\alpha r} \leq m 2^{2\alpha r}$ bits of the truth table of $h$.

Then, given a distinguisher $D$ with advantage $\beta$, we can approximate $f$ with the functions $(A, \leak)$ described below.

On input $x, z = f(x)$, $\leak$ proceeds as follows (denote $h(i) = z_i$):
\begin{enumerate}
    \item Evaluates $\I = \GenDesign(d,\alpha)$, with $d$ chosen the same way as the PRG.
    \item Chooses a random $j \in [m]$ and a random $y_{[d]\setminus I_j}$, storing them on its work tape.
    \item Writes $j$ and $y_{[d]\setminus I_j}$ to the output tape.
    \item For $0 < i < j$ and for each value of $y$ obtained by ranging over possible values of $y_{I_j}$ (since $y_{[d] \setminus I_j}$ is already fixed), $\leak$ writes $h(y_{I_i})$ to the output tape.
    \item Writes random bits $b, w_j$ and string $w_{[m] - [j]}$ to the output tape.
\end{enumerate}

Then, on input $x, w = \leak(x, f(x))$, $A$ proceeds as follows:
\begin{enumerate}
    \item $A$ evaluates $\I = \GenDesign(d,\alpha)$, with $d$ chosen the same way as the PRG.
    \item Defines $b$, $j$, $y_{[d] \setminus I_j}$, $w_j$, and $w_{[m]\setminus[j]}$ to be those read from its input tape.
    \item For each $y_{I_j} = k$ in $[n]$ yielding a full string $y$,
    \begin{enumerate}
        \item Computes $H_{j-1}$, with bit $i < j$ found by looking in position $k$ in the $i$th table output by $\leak$ in step 4 and $i \geq j$ found by taking the appropriate bit of $w$, also output by $\leak$.
        \item Computes $\distinguish(H_{j-1}) \oplus b \oplus w_j$ and writes it to the output tape.
    \end{enumerate}
\end{enumerate}

We will show that the result is a pair which runs in the necessary time and space bounds and, with high probability, is a good enough approximation of $f(x)$ to achieve contradiction.

\proofsubparagraph{Space-bound}
$\leak$ step 1 requires space $S_\GenDesign = \mathcal{O}(d) = O\q(\frac{\log n}{\alpha})$. Since it is locally computed by future steps, this becomes $\mathcal{O}(S_\GenDesign)$. The second step requires space $d + \log m$, and uses the local computation of step 1. The third step requires space $\log (\log m + d)$, since it is simply iteration and printing. The fourth step uses space $\log m + d$ again since it iterates over indices $i$ and values of $y$. For each value, it simply performs a lookup into the input, which requires at most space equal to the size of the index, bounded by $d + 1$. Finally, step $5$ requires space $\log m$, since it needs to keep track of the number of bits it needs to write.

Thus, the overall space used by $\leak$ is $\mathcal{O}(d) + \mathcal{O}(\log n) + \mathcal{O}(\log m) = \mathcal{O}(\log n + \log \alpha) = \mathcal{O}(\log n)$.

For $A$, step 1 again requires space $\mathcal{O}(S_\GenDesign)$. Step 2 is simply a definition and thus requires no space. Step 3 requires $\log n$ bits of iteration overhead plus the space used by 3.(a) and 3.(b). 3.(a) is simply indexing into the input based on $k$ and so requires space $\log n$. This is the locally recomputed by step 3.(b), so 3.(a) requires space $\mathcal{O}(\log n)$. Finally, 3.(b) requires space $C \log m$.

Thus, $A$ and $\leak$ run in space $c_2 \log n$ for a constant $c_2$ dependent only on $C$.

\proofsubparagraph{Time-bound}
$\GenDesign$ is a deterministic logspace algorithm, and thus runs in time $T_\GenDesign = \poly(n)$. In order compose with it, we need to run it as many times as a bit of it is needed, so it becomes a multiplicative factor on the rest of the run time.

For $\leak$, for steps 2 and 3, choosing random strings requires time linear in their lengths, since one simply reads from the random tape and copying them to the tape is also linear. Thus, the time complexity of this step is $\mathcal{O}(\log n + \log m)$. Since $j \leq m$ and $|y_{I_j}| = \log n$, there are $\leq n \cdot m$ iterations in step 4, and each step is simply indexing into a string of length $n$, so step 4 has time complexity $\mathcal{O}(n^2 \cdot m)$. Step 5 is again linear, requiring $\mathcal{O}(m)$ time to write $\leq m$ bits to the tape.

Thus, $\leak$ runs in time $n^{c_{1,1}}$ for some constant $c_{1,1}$.

Denote by $T_D$ the time complexity of the distinguisher. Note that $T_D = \poly(n) \cdot \mathcal{O}(2^{C \log m}) = \poly(n)$.

For $A$, step 1 is polynomial as above and step 2 is a no-op. The iteration in step 3 is over $n$ terms. Step 3.(a) is $\mathcal{O}(m\cdot n)$ since it's simply a look up into the output of $\leak$ which has sub-linear length (see contradiction of hardness below) for each of $m$ bits. This is done each time the distinguisher in step 3.(b) needs a bit, so step 3's overall time complexity is $\mathcal{O}(T_D \cdot m \cdot n^2)$.

Thus, $A$ runs in time $n^{c_{1,2}}$ for some $c_{1,2}$.

Setting $c_1 = \max(c_{1,1},c_{1,2})$ gives the desired time bound.

\proofsubparagraph{Contradiction of hardness}
First, note that $\leak$ outputs at most $\log m + d - r + m + m 2^{2\alpha^2 d}$ bits. The $\log m$ and $d - r$ terms are logarithmic and thus fall within any polynomial bound. The constraints on $d$ and $\alpha$, chosen such that $d \alpha = \log n$ and $2^{\frac{\alpha^4d}{5}} = m$ tell us that for the $m 2^{2\alpha^2 d}$ term, we have that 

$$2\alpha^2d = \frac{10}{\alpha^2} \log m \quad \text{so} \quad 2^{2\alpha^2 d} = m^{\frac{10}{\alpha^2}}$$

Thus, $m 2^{2\alpha^2 d} = m^{\frac{10}{\alpha^2} + 1}$, which is

$$m^{\frac{10}{\alpha^2} + 1} = \q(n^{\frac{\alpha^3}{5}})^{\frac{10}{\alpha^2} + 1} = n^{2\alpha + \frac{\alpha^3}{5}}$$

Thus, if $\alpha$ is small enough (e.g., $\frac{1}{3}$), then for $\varepsilon = 2\alpha + \frac{\alpha^3}{5} < 1$, we have that $|\leak(x,f(x))| \leq n^\varepsilon$.

Observe that \eqref{eq:good_choice} is satisfied with probability $\frac{\beta}{2m}$ by $\leak$'s random choice of $j$, $y_{[d]\setminus I_j}$, and $w_{[m]\setminus[j]}$. Furthermore, with probability $1/4$, $\leak$'s random choices of $b$ and $w_j$ are the correct choices such that \eqref{eq:good_approximation} holds. Thus, with probability $\frac{\beta}{8m}$, we have that \eqref{eq:good_approximation} holds for our choices. 

Consider the case where it does. Note that the randomness in \eqref{eq:good_approximation} is over the argument to $h$, which is the index into $f(x)$. Then, when $A$ evaluates $\distinguish(H_{j-1}) \oplus b \oplus w_j$ for each index, we conclude that at least a $\frac{1}{2} + \frac{\beta}{2m}$ fraction of bits output by $A$ agree with $f(x)$. We thus have that

$$\P[d_H(A(x, \leak(x, f(x))), f(x)) \leq \q(\frac{1}{2} - \frac{\beta}{2m})n] \geq \frac{\beta}{8m}$$

However, note that $\frac{\beta}{8m} > \inv{n}$ for sufficiently large $n$ (since $m$ is polynomially smaller than $n$) and $\inv{2} - \frac{\beta}{2m} < \inv{2} - \inv{m^2}$ for any constant $\beta$ and sufficiently large $n$. This contradicts the hardness assumption of $f$, so such a distinguisher can't exist and we must have that $G$ is a $C \log m$-secure PRG.
\end{proof}

\subsection{Derandomization of search problems using the PRG}
Now, we prove that the PRG is sufficient. This formulation is similar to work by Goldreich~\cite{goldreich_2011} which was seen in Liu and Pass's~\cite{liu_2022} reproduction.
\begin{lemma}[Derandomization of search problems using the PRG]
    \label{lemma:prg-derand}
    The existence of a $3\log m$-secure $(\poly(m), \mathcal{O}(\log m))$, logspace-computable PRG implies (2) of main theorem (\ref{thm:main}).
\end{lemma}

As is typical for derandomization using a PRG, we will need the ability to pad strings. We will treat functions $\pad(x, k)$ and $\unpad(x')$ as having ``negligible'' complexity. The exact definition and proof of this can be seen in lemma \ref{lemma:pad}, but we effectively use an expanded alphabet (encoded as pairs of bits) to add an ``pad'' character.

We separate our proof of \ref{lemma:prg-derand} into two parts: derandomizing $\prBPdL$ using a PRG and then derandomizing $\prBPdL$ search problems using both $\prBPdL = \prL$ and the PRG.

\begin{lemma}[Derandomization of promise problems]
    \label{lemma:prg-promise-derand}
    The existence of a $3\log m$-secure $(\poly(m), \mathcal{O}(\log m))$, logspace-computable PRG implies $\prBPdL = \prL$.
\end{lemma}
\begin{proof}[Proof sketch]
    This is an adaptation of the time-bounded derandomization construction as seen in Goldreich's and Liu and Pass's work. We have included it in detail in appendix \ref{appendix:lemma:prg-derand} to show that there are no hiccups related to logarithmic space. (In particular, the unpadded case becomes time $\mathcal{O}(k)$ and space $\log k + \mathcal{O}(1)$, and the padding exponent becomes the maximum of the exponent on time and the multiplier on space.)
\end{proof}

\begin{proof}[Proof of \ref{lemma:prg-derand}]
    Suppose $G$ is an $3 \log m$-secure $(n,d)$ PRG with parameters $n = m^\theta$ for some $\theta$ and $d = \mathcal{O}(\log m)$. By lemma \ref{lemma:prg-promise-derand}, we have that $\prBPdL = \prL$.
    
    Suppose $\RYES, \RNO$ is a $\prBPdL$ search problem with verifier $V(x,y;\omega)$ and finder $F(x;\gamma)$ where $|\gamma| \leq t(|x|)$. Observe that we can derandomize $V$ to $V(x,y)$ (since it is a $\prBPdL$ decision problem) such that $(x,y) \in \RNO \implies V(x,y) = 0$ and $(x,y) \in \RYES \implies V(x,y) = 1$.
    
    First we will perform this derandomization for search problems where $F(x; \gamma)$ has time complexity $t(k) = \mathcal{O}(k)$ and the random procedure $M(x;\gamma) = V(x, F(x; \gamma))$ has space complexity $\log k + \mathcal{O}(1)$. Define $m = t(k)$ as the maximum amount of randomness used.
    
    Consider the following deterministic procedure $F'$: On input $x$,
    \begin{enumerate}
        \item For each seed $s \in \binstr^d$,
        \begin{itemize}
            \item Compute $b \gets V(x, F(x; \generator[x'](s)))$, where $x' = \pad(x,n)$.
            \item If $b = 1$, output $F(x; \generator[x'](s))$.
        \end{itemize}
        \item Output the empty string.
    \end{enumerate}
    
    Note that this procedure operates in logspace since iterating seeds requires space $d = \mathcal{O}(\log m) = \mathcal{O}(\log k)$ and step 3.(a) and 3.(b) are logspace compositions (a more detailed treatment can be see in the proof of \ref{lemma:prg-promise-derand}).
    
    Observe also that if this procedure terminates within the loop, it outputs a string $y$ such that $V(x,y) = 1$ and thus such that $(x,y) \notin \RNO$. It remains to show that this procedure terminates within the loop on all but finitely many inputs (since we can hard-code any finite number of misbehaving cases).
    
    Suppose the contrary. Then there exist infinitely many inputs for which, 
    
    $$\forall s \in \binstr^d ~~~V(x, F(x; G_{x'}^m(s))) = 0$$ 
    
    Further, by definition of a finder, $\P_\gamma[(x, F(x; \gamma)) \in \RYES] \geq \twothirds$. Thus,
    
    $$\P_s[V(x, F(x; G_{x'}^m(s))) = 1] = 0 \qand \P_\gamma[(x, F(x; \gamma)) \in \RYES] \geq \twothirds$$
    
    This tells us that $D(1^m, x', r) = M(\unpad(x); r) = M(x;r)$ is a distinguisher with advantage $\beta \geq \twothirds$ which works for infinitely many values of $x'$ and runs in space $2 \log m$ (since it is the truncation machine composed with $M$; again, a more detailed treatment can be seen in the proof of \ref{lemma:prg-promise-derand}). This contradicts the security of the generator.
    
    Thus, $(x, F'(x)) \notin \RNO$ for all but finitely many $x$ (which can be hard-coded).
    
    Now consider an arbitrary $\prBPdL$ search problem, $(\RYES, \RNO)$ with derandomized verifier $V$ and finder $F(x;\gamma)$ such that the time complexity of $F$ is $\mathcal{O}(k^a)$ and the space complexity of $V(x, F(x; \gamma))$ is $b \log k + \mathcal{O}(1)$. Again, let $c = \max(a, b)$ and define $x' = \pad(x,k^c)$. We create $\RYES_\pad = \set{(x',y): (x,y) \in \RYES}$ and $\RNO_\pad = \set{(x',y): (x,y) \in \RNO}$.
    
    Then $V_\pad(x',y) = V(\unpad(x'),y)$, $F_\pad(x') = F(\unpad(x'))$ satisfy $$M_\pad(x';r) = V_\pad(x',F_\pad(x';r))$$ being a machine with space complexity $\log k' + \mathcal{O}(1)$ and $F_\unpad$ having time complexity $\mathcal{O}(k')$. Thus, we have deterministic algorithm $F'_\pad(x')$ such that $(x',F'_\pad(x')) \notin \RNO_\pad$.
    
    We can then define $F'$ as the function which on input $x$ computes $F'_\pad(\pad(x,k^c))$. Observe that if $(x,F'(x)) \in \RNO$ then $(x',F'_\pad(x')) = (x',F'(x)) \in \RNO_\pad$, which contradicts $F'_\pad$ being a derandomized finder for $\RNO_\pad$.
    
    Thus, we have achieved the desired partial derandomization of $\prBPdL$ search problems.
\end{proof}

\section{Derandomization implies existence of hard function}
\begin{theorem}[Derandomization implies existence]
    \label{thm:derandomization-to-existence}
    Suppose that for any $\prBPdL$ search problem $(\RYES, \RNO)$, we can create a deterministic algorithm $F'$ such that for any $x \in S_R$, we have that $(x,F'(x)) \notin \RNO$.

    We will show that for any constants $c_1,c_2$ there exists a logspace computable function $f$ such that $f$ is almost-all-input $(T, S, \ell, d) = \q(n^{c_1}, c_2\log n, n^\varepsilon, \q(\frac{1}{2} - \frac{1}{m^2})n)$-leakage-resilient average hard where $\varepsilon = 2\alpha + \frac{\alpha^3}{5}$ and $m = n^{\frac{\alpha^3}{5}}$ for $\alpha = \frac{1}{3}$.
\end{theorem}

The ``Random is hard'' lemma is similar to claim 1 in section 3.1 of Liu and Pass's paper~\cite{liu_2022} but needs a much stronger claim than in the time-bounded setting, corresponding to the use of an average-case hardness assumption.

\subsection{Random is hard}
\begin{lemma}[Random is hard]
    \label{lemma:random_is_hard}
    For \emph{any} probabilistic algorithms $(A, \leak)$ and for all $n \in \N$, $x \in \binstr^n$, it holds that
    
    $$\P[|\leak(x,r)| \leq \ell \land d_H(A(x, \leak(x, r)), r) < d] \geq \inv{2n}$$

    with probability at most $n \cdot 2^{n\cdot (H(d/n) - 1) + \ell + \mathcal{O}(1)}$ (where $H(p)$ is the binary entropy function $H(p) = p \log \inv{p} + (1 - p) \log \inv{1 - p}$) over $r \sim \U_n$.
\end{lemma}

\begin{proof}
    Consider any $n \in \N$, $x \in \binstr^n$, and any probabilistic algorithm $A$. We will show that for any \emph{deterministic} function $\leak'$ that outputs at most $\ell$ bits,
    \begin{equation}
        \label{eq:random_is_hard}
        \P_{r \sim \U_n}[\P[d_H(A(x, \leak'(x, r)), r) < d] \geq \inv{2n}] \leq n \cdot 2^{n\cdot (H(d/n) - 1) + \ell + \mathcal{O}(1)}
    \end{equation}
    The claim for any probabilistic algorithm $\leak$ follows immediately by letting $\leak'$ be $\leak$ with the best random tape fixed.
 To show \eqref{eq:random_is_hard}, consider the set of ``bad'' $r$'s defined as
    $$B = \set{r \in \binstr^n: \exists w \in \binstr^\ell~\mathrm{s.t.}~\P[d_H(A(x, \leak(x, r)), r) < d] \geq \inv{2n}}$$
    For any $r$, if there exists a $\leak'$ such that $|\leak'(x, r)| \leq \ell$ and $\P[d_H(A(x, \leak'(x, r)), r) < d] \geq \inv{2n}$, then $r \in B$. Thus, the probability that $A(x,\leak(x,r))$ successfully approximates $r$ is bounded by the probability that $r \in B$.

    We now bound $B$. Observe that for a given choice of $w \in \binstr^{\leq \ell}$, there exists $2n$ Hamming balls (which, in the worst case, do not overlap) of radius $d$ contained in $B$, since there are at most $2n$ Hamming balls which $g(x, w)$ can occupy with probability $\geq \inv{2n}$.

    Finally, note that a Hamming ball of radius $d$ on strings of length $n$ has volume
    $$2^{nH(d/n) - \inv{2}\log n + \mathcal{O}(1)} \leq 2^{nH(d/n) + \mathcal{O}(1)}. \quad \text{Thus,} \quad |B| \leq 2^{\ell + 1} \cdot 2n \cdot 2^{nH(d/n) + \mathcal{O}(1)}$$
    We can roll the $2$ in $2n$ into the $\mathcal{O}(1)$ in the final exponent, along with the $+1$ in $\ell + 1$.

    Thus, the probability that a randomly chosen $r$ satisfies the condition in \eqref{eq:random_is_hard} is at most
    $$\frac{|B|}{2^n} \leq n \cdot 2^{n(H(d/n) - 1) + \ell + \mathcal{O}(1)}$$
    \vspace{-1em}
\end{proof}

\subsection{Construction of hard function}

\begin{proof}[Proof of theorem \ref{thm:derandomization-to-existence}]
We can construct a $\prBPdL$ search problem to represent computing our desired hard problem. This search problem is similar to the corresponding search problem from~\cite{liu_2022} with the changes corresponding directly to the differences between the time-bounded worst-case version of leakage-resilient hardness and the space-bounded average-case version presented in this work.
\proofsubparagraph{The search problem}
Note that, since we consider only efficiently and uniformly computable $(A, \leak)$, a given pair has some constant length. We can thus consider only pairs $(A, \leak)$ with descriptions of length at most $\log n$. For any given attacker pair, then, we will only skip it under our search problem in a finite number of cases, which is covered by the almost-all-input hardness.

Let $\RYES$ be a binary relation such that $(x, r) \in \RYES$ if both of the following are true:
\begin{enumerate}
    \item $|x| = |r|$.
    \item For all probabilistic $(A, \leak)$ such that $|A|, |\leak| \leq \log n$ (where $|A|$ (resp. $\leak$) refers to the length of the description of the program $A$ in some arbitrary encoding scheme), it holds that
    \begin{equation}
        \label{eq:yes_membership}
        \P[|\leak'(x,r)| \leq \ell \land d_H(A'(x, \leak'(x, r)), r) < d] < \inv{2n}
    \end{equation}

    where $A'$ and $\leak'$ are time-space-truncated versions of $A$ and $\leak$ which are only executed until they run for $n^{c_1}$ steps or try to use more than $c_2 \log n$ space (where $c_1,c_2$ are from theorem \ref{thm:derandomization-to-existence}).
\end{enumerate}

Let $\RNO$ be a binary relation such that $(x,r) \in \RNO$ if $|x| \neq |r|$ or for \emph{at least one pair} of $(A, \leak)$ (within $\log n$ size bound), \eqref{eq:yes_membership} doesn't hold with a $\inv{n}$ lower bound instead of an $\inv{2n}$ upper bound.

We then show this is a $\prBPdL$ search problem by finding a verifier and finder.

\proofsubparagraph{Verifier}
On input $(x,r)$, the verifier immediately rejects if $|x| \neq |r|$. Otherwise, it enumerates all probabilistic machines $(A, \leak)$ such that $|A|, |\leak| \leq \log n$. For each one, $V$ estimates the value
$$p_{A,\leak} = \P[|\leak'(x,r)| \leq \ell \land d_H(A'(x,\leak'(x,r)), r) < d]$$

by simulating $A'(x,\leak'(x,r))$ polynomially many times and recording whether the Hamming distance is at most $d$. If during this process, for any $(A, \leak)$, we have that the sample estimate $\hat{p}_{A,\leak}$ exceeds $\frac{3}{4n}$, $V$ outputs $0$ and terminates. Otherwise, the value $p_{A,\leak}$ is discarded for this attacker pair and estimated for the next attacker pair. If we iterate every attacker pair without terminating, we output $1$ and terminate. By the Chernoff bound and the union bound, $V$ will accept with high probability if $(x,r) \in \RYES$ and reject with high probability if $(x,r) \in \RNO$.

This step is exactly where two-way randomness becomes necessary. We can compute $A'(x, \leak'(x,r))$ in logspace since with two-way randomness of $\BPdL$, we can compose functions. The rest is logspace, since computing Hamming distance can be done in a read-once manner by tallying number of agreeing bits, which requires $\log n$ bits, and keeping a tally of successes vs. total count represents the success probability in $\log \mathsf{num\_trials}$ space. Since only a polynomial number of trials is necessary, this is all logspace.

We can compute it in polynomial time, since we only need polynomially many samples and each of $A'$ and $\leak'$ is time truncated to a polynomial time-bound.

\proofsubparagraph{Finder}
On input $x$, the finder outputs a random string of the same length. Observe that for any fixed $x \in \binstr^n$, by the ``random is hard'' lemma (\ref{lemma:random_is_hard}), and a union bound over choices of $(A, \leak)$, we conclude that $F(x)$ outputs an invalid witness with probability at most
$$n^2 \cdot n \cdot 2^{n(H(d/n) - 1) + \ell + \mathcal{O}(1)}$$
When $\ell = n^\varepsilon = n^{2\alpha + \frac{\alpha^3}{5}}$ and $\frac{d}{n} = \frac{1}{2} - \frac{1}{m^2}$ with $m = n^{\frac{\alpha^3}{5}}$ and $\alpha = \frac{1}{3}$, this converges to 0, and thus satisfies the search problem requirement of error probability less than $\onethird$ (except for perhaps some finite number of small samples where answers can be hard-coded).

\proofsubparagraph{Computing the hard function}
By the derandomization assumption, there exists a deterministic, logspace function $F'$ such that $(x,F'(x)) \notin \RNO$. 

Note that $F'(x) \notin \RNO$, so $|F'(x)| = |x|$ and for all machines $(A, \leak)$ which run within the time and space bounds, \eqref{eq:yes_membership} (with $\inv{2n}$ replaced with $\inv{n}$) holds for all but the finitely many strings too short to include the descriptions of $A, \leak$.

Therefore, by construction, $f = F'$ is a logspace computable $(n^{c_1}, c_2 \log n,\ell,d)$-leakage-resilient average hard function as desired. 
\end{proof}

\section{Discussion}

We suspect it is possible to reduce the derandomization part of our assumption to $\prBPdL = \prL$, such that we do not explicitly need the separate partial derandomization of $\prBPdL$ search problems. However, we were not able to identify such a reduction within the duration of this research.

\bibliography{bib}
\clearpage
\appendix

\section{Proof of lemmas for derandomization using PRG}\label{appendix:lemma:prg-derand}
\begin{lemma}[Efficient padding]
    \label{lemma:pad}
    There exists a function $\pad(x,k)$ (note that $k$ is not passed in unary) which has three critical properties:
    \begin{enumerate}
        \item $|\pad(x,k)| \geq k$.
        \item $\pad$ has space complexity equal to $\mathcal{O}(|k|)$ (where $|k| = \log k$ is the bits to write down the numeric value of $k$).
        \item The inverse $\unpad$ of $\pad$ (which produces $x$ from $\pad(x,k)$) has online space complexity $\mathcal{O}(1)$ and online time complexity $\mathcal{O}(|x|)$ (both independent of $k$), even if $k$ is not known.
    \end{enumerate}
    Furthermore, if $\unpad$ is applied to the contents of an input tape directly, it can transparently provide access to $x$ with $\mathcal{O}(1)$ space overhead, even with two-way seeks.
\end{lemma}
\begin{proof}
    We will explicitly construct $\pad$, its generic inverse, and its tape-emulator inverse.
    $\pad$ proceeds as follows: on input $k, x$, $\pad$ initializes a $\log k$ bit counter $C$ to $k$, then it traverses $x$, writing $1x_i$ for each $i \leq |x|$, setting $C$ to $k - 2|x|$ in the process. It then pads the end with $C$ $0$'s, if $C > 0$.
    
    The first two properties follow from the construction.
    
    The third property is demonstrated by explicitly providing the algorithm for $\unpad(x')$ where $x' = \pad(x,k)$ for any $k$. $\unpad$ simply prints each odd (indexed from 0) bit of $x'$ until the preceding even bit is $0$ or it runs out of bits
    
    This algorithm reads at most $2|x| + 1$ bits from the tape in a read-once fashion and it only needs to keep track of the parity, so it has the desired property.
    
    Finally, observe that we can wrap the input tape by simply doubling all moves and output a blank when the bit the left of the head is a 0.
\end{proof}

\begin{proof}[Proof of \ref{lemma:prg-promise-derand}]
    Suppose $G$ is a $3 \log m$-secure $(n,d)$ PRG with parameters $n = m^\theta$ for some $\theta$ and $d = \mathcal{O}(\log m)$.
    
    First, we will show that we can derandomize $\prBPdL$ decision problems which can be solved in probabilistic time $\mathcal{O}(k)$ and space $\log k + \mathcal{O}(1)$ for inputs of length $k$. Suppose that for some promise problem $(\LYES, \LNO)$, $M(x;\gamma)$ is a machine which runs in probabilistic (over $\gamma$) time bounded by $t(k) = \mathcal{O}(k)$ for a logspace computable function $t$ and space $s(k) = \log k + \mathcal{O}(1)$ and has the properties that
    $$\forall x \in \LYES,~\P_\gamma[M(x;\gamma) = 1] \geq \twothirds \qand \forall x \in \LNO,~\P_\gamma[M(x;\gamma) = 1] \leq \onethird$$
    
    Using $G$, we will construct a deterministic machine $M'$ such that $M'$ runs in space $\mathcal{O}(\log k)$ and $\forall x \in \LYES,~M(x) = 1$ and $\forall x \in \LNO,~M(x) = 0$.
    
    We need up to $t(k)$ bits of randomness, so define $m = t(k)$ and $n = m^\theta$. Since our targeted generator requires a target string of length $n$, define $x' = \pad(x, n)$. Since $n = t(k)^\theta$ can be computed in logspace, $x'$ is logspace computable by the efficient padding lemma (\ref{lemma:pad}). Observe that we can treat $M(x;\gamma)$ as a logspace (in the first argument) deterministic function of $x$ and $\gamma$. Since $M$ runs in time $t(k) = \mathcal{O}(k)$, this means it can read at most $t(k)$ bits of $\gamma$, which means it is a deterministic logspace function of both its arguments together.
    
    Thus, $M'$ can proceed as follows:
    \begin{enumerate}
        \item Initialize a $d$-bit counter $C$.
        \item For every seed $s \in \binstr^{d}$,
        \begin{enumerate}
            \item Compute $b \gets M(x;\generator[x'](s))$.
            \item If $b = 1$, increment $C$.
        \end{enumerate}
        \item Accept iff $C > \frac{1}{2}2^d$, otherwise reject.
    \end{enumerate}
    
    Observe that the currently considered seed and counter both occupy space $d = \mathcal{O}(\log m) = \mathcal{O}(\log t(k)) = \mathcal{O}(\log k)$. Step 2.(a) is composition of deterministic logspace functions, which can be done in logspace. The comparison and increments can all be done in logspace.
    
    Thus, this procedure runs in space $\mathcal{O}(\log k$).
    
    It remains to show that this must correctly derandomize all sufficiently long inputs $x$ (if there is some finite set of $x$'s which this does not derandomize correctly, we say that $M'$ simply hard-codes the correct answers to those $x$'s).
    
    Suppose there are infinitely many values of $x$ for which this procedure produces the wrong output. Then, the procedure $D(1^m, x', r) = M(\unpad(x'); r) = M(x;r)$ serves as a distinguisher with advantage $\beta \geq \frac{1}{6}$, since we know that
    
    $$\P_{\gamma \sim \U_m}[M(x;\gamma) = b] \geq \twothirds \qand \P_{s \sim \U_d}[M(x;\generator[x'](s)) = b] \leq \inv{2}$$
    
    where $b$ is the correct output for $x$.
    
    Furthermore, $M(x; r)$ runs in space $\log k + \mathcal{O}(1)$ and thus the machine $D$ runs in space $\log k + \mathcal{O}(1)$ (since $\unpad$ applied to the input tape runs in space $\mathcal{O}(1)$), which is at most space $\log m + \mathcal{O}(1) \leq 2\log m$, which contradicts the security of $G$. Thus, $M'$ is a derandomized version of $M$ (the value of $M'$ on $x \notin \LYES \cup \LNO$ is unspecified).
    
    We can then use this to derandomize any $\prBPdL$ problem via a padding argument. Given a $\prBPdL$ problem $L = (\LYES, \LNO)$ solvable in time $\mathcal{O}(k^a)$ and space $b \log k + \mathcal{O}(1)$, let $c = \max(a, b)$ and define $L_\pad = (\LYES_\pad, \LNO_\pad)$ such that $x' \in \LYES_\pad$ iff $\exists x \in \LYES$ such that $x' = \pad(x, k^c)$ and similarly for $\LNO_\pad$. Then, $L_\pad$ can be solved in time $\mathcal{O}(k')$ and space $\log k' + \mathcal{O}(1)$ by using the tape variant of $\unpad$ and the verifier and finder for $L$.
    
    Thus, we can derandomize $L_\pad$. But if we do so, we can also construct a deterministic logspace machine which, on an instance $x \in \LYES \cup \LNO$, first pads $x$ to $x'$, then runs the deterministic solver for $L_\pad$ and outputs its result.
    
    Thus, we conclude that if such a generator exists, $\prBPdL = \prL$.
\end{proof}

\end{document}